\documentclass[twoside,leqno]{article}
\pdfoutput=1
\usepackage[letterpaper]{geometry}

\usepackage{ltexpprt}
\usepackage{hyperref}
\usepackage{thm-restate}

\def\showauthornotes{0}

\def\showdraftbox{0}

\usepackage{amsfonts}
\usepackage{amsmath}
\usepackage{amssymb}
\usepackage{amstext}
\usepackage[noend]{algpseudocode}
\usepackage{color}
\usepackage{bbm}
\usepackage{bm}
\usepackage{graphicx}
\usepackage{algorithmicx}
\usepackage{algorithm}
\usepackage{thmtools}
\usepackage[capitalize]{cleveref}

\newcommand{\eps}{\varepsilon}

\newenvironment{fminipage}%
  {\begin{Sbox}\begin{minipage}}%
  {\end{minipage}\end{Sbox}\fbox{\TheSbox}}

\def\abs#1{\left| #1 \right|}
\newcommand{\norm}[1]{\ensuremath{\left\lVert #1 \right\rVert}}

\newcommand\rea{\mathbb R}

\newcommand{\marginlabel}[1]%
{\mbox{}\marginpar{\it{\raggedleft\hspace{0pt}#1}}}

\DeclareMathOperator{\poly}{poly}

\definecolor{Mygray}{gray}{0.8}

 \ifcsname ifcommentflag\endcsname\else
  \expandafter\let\csname ifcommentflag\expandafter\endcsname
                  \csname iffalse\endcsname
\fi

\ifnum\showauthornotes=1
\newcommand{\todo}[1]{\colorbox{Mygray}{\color{red}#1}}
\else
\newcommand{\todo}[1]{}
\fi

\ifnum\showauthornotes=1
\newcommand{\Authornote}[2]{{\sf\small\color{red}{[#1: #2]}}}
\newcommand{\Authoredit}[2]{{\sf\small\color{red}{[#1]}\color{blue}{#2}}}
\newcommand{\Authorcomment}[2]{{\sf \small\color{gray}{[#1: #2]}}}
\newcommand{\Authorfnote}[2]{\footnote{\color{red}{#1: #2}}}
\newcommand{\Authorfixme}[1]{\Authornote{#1}{\textbf{??}}}
\newcommand{\Authormarginmark}[1]{\marginpar{\textcolor{red}{\fbox{
#1:!}}}}
\else
\newcommand{\Authornote}[2]{}
\newcommand{\Authoredit}[2]{}
\newcommand{\Authorcomment}[2]{}
\newcommand{\Authorfnote}[2]{}
\newcommand{\Authorfixme}[1]{}
\newcommand{\Authormarginmark}[1]{}
\fi

\newlength{\pgmtab}  
\let\originalleft\left
\let\originalright\right
\renewcommand{\left}{\mathopen{}\mathclose\bgroup\originalleft}
  \renewcommand{\right}{\aftergroup\egroup\originalright}

\def\abs#1{\left|#1  \right|}

\def\norm#1{\left\| #1 \right\|}

\newcommand\ppsi{\boldsymbol{\mathit{\psi}}}

\newcommand\PPi{\boldsymbol{\Pi}}

\newcommand\ppi{\boldsymbol{\pi}}

\newcommand\pp{\boldsymbol{\mathit{p}}}

\newcommand\xx{\boldsymbol{\mathit{x}}}

\renewcommand\AA{\boldsymbol{\mathit{A}}}

\newcommand\DD{\boldsymbol{\mathit{D}}}

\newcommand\NN{\boldsymbol{\mathit{N}}}
\newcommand\MM{\boldsymbol{\mathit{M}}}
\newcommand\LL{\boldsymbol{\mathit{L}}}

\newcommand\Otil{\widetilde{O}}

\def\qedsketch{\ifmmode\Box\else{\unskip\nobreak\hfil
\penalty50\hskip1em\null\nobreak\hfil$\Box$
\parfillskip=0pt\finalhyphendemerits=0\endgraf}\fi}

\newlength{\tpush}
\newcommand{\handout}[5]{
   \noindent
   \begin{center}
   \framebox{ \vbox{ \hbox to \textwidth { {\bf \coursenum\ :\  \coursename} \hfill #5 }
       \vspace{3mm}
       \hbox to \textwidth { {\Large \hfill #2  \hfill} }
       \vspace{1mm}
       \hbox to \textwidth { {\it #3 \hfill #4} }
     }
   }
   \end{center}
   \vspace*{4mm}
   \newcommand{\lecturenum}{#1}
   \addcontentsline{toc}{chapter}{Lecture #1 -- #2}
}

\ifnum\showdraftbox=1

\else

\fi

\allowdisplaybreaks

\renewcommand{\epsilon}{\varepsilon}

\usepackage[
  backend=biber,
  backref=true,
  backrefstyle=none,
  date=year,
  doi=false,
  giveninits=false,
  hyperref=true,
  maxbibnames=20,
  maxnames=99,
  minnames=9,
  maxalphanames=10,
  sortcites=false,
  style=numeric,
  url=false, 
]{biblatex}
\newcommand{\lawrence}{\Authornote{Lawrence}}

\begin{document}
\title{A New Approach to Estimating Effective Resistances and Counting Spanning Trees in Expander Graphs
}
\author{Lawrence Li\\Department of Computer Science\\University
of Toronto\\lawrenceli@cs.toronto.edu
\and 
Sushant Sachdeva\\Department of Computer Science\\University
of Toronto\\sachdeva@cs.toronto.edu
}

\maketitle
\begin{abstract}
  We demonstrate that for expander graphs, for all $\eps > 0,$ there exists a data structure of size $\Otil(n\eps^{-1})$ which can be used to return $(1 + \eps)$-approximations to effective resistances in $\Otil(1)$ time per query. 
  Short of storing all effective resistances, previous best approaches could achieve $\Otil(n\eps^{-2})$ size and $\Otil(\eps^{-2})$ time per query by storing  Johnson–Lindenstrauss vectors for each vertex, or  $\Otil(n\eps^{-1})$ size and  $\Otil(n\eps^{-1})$ time per query by storing a spectral sketch.

  Our construction is based on two key ideas: 1) $\eps^{-1}$-sparse, $\eps$-additive approximations to $\bm{\sigma_u}$ for all $u$, vectors similar to $\DD\LL^+\bm{1_u}$, can be used to recover $(1 + \eps)$-approximations to the effective resistances, 2) In expander graphs, only $\Otil(\eps^{-1})$ coordinates of $\bm{\sigma_u}$ are larger than $\eps.$
  We give an efficient construction for such a data structure in $\Otil(m + n\eps^{-2})$ time via random walks.
  This results in an algorithm on expander graphs for computing $(1+\eps)$-approximate effective resistances for $s$ vertex pairs that runs in $\Otil(m + n\eps^{-2} + s)$ time, improving over the previously best known running time of $m^{1 + o(1)} +  (n + s)n^{o(1)}\eps^{-1.5}$ for $s = \omega(n\eps^{-0.5}).$
  
  We employ the above algorithm to compute a $(1+\delta)$-approximation to the number of spanning trees in an expander graph, or equivalently, approximating the (pseudo)determinant of its Laplacian in $\Otil(m + n^{1.5}\delta^{-1})$ time. This improves on the previously best known result of $m^{1+o(1)} + n^{1.875+o(1)}\delta^{-1.75}$ time, and matches the best known size of determinant sparsifiers.
\end{abstract}

\thispagestyle{empty}

\newpage

\setcounter{page}{1}

\section{Introduction}
\paragraph{Estimating Laplacian Determinants.}
Given a graph $G$ with $n$ vertices and $m$ edges,
what is the running-time complexity of estimating the number of spanning trees in $G$? 
In the 19th century, the celebrated matrix-tree theorem of Kirchhoff~\cite{Kirchhoff47} established that the number of spanning trees is equal to $1/n$ times the product of the non-trivial eigenvalues of the graph Laplacian, or equivalently the (pseudo)determinant of the Laplacian.\footnote{The usual determinant for the Laplacian is 0 since it has a trivial eigenvalue of 0. In this paper, we overload the notation Laplacian determinant to denote the pseudodeterminant, the product of all non-zero eigenvalues}.
This implies an $O(n^\omega)$ time algorithm for exactly computing the number of spanning trees in $G.$

This result naturally extends to estimating the total weight of all spanning trees for graphs with non-negative weights $w$ on the edges, where the weight of a tree $w(T)$ is defined as the product of the weights of the edges in the tree $\prod_{e \in T} w(e).$

The first result to improve on the matrix multiplication time bound was by \textcite{DurfeePPR17}, giving an $\Otil(n^2\delta^{-2})$ time\footnote{The $\Otil(\cdot)$ notation hides $\poly(\log n)$ factors.} algorithm for returning a $(1+\delta)$-approximation to the Laplacian determinant.
 \textcite{ChuGPSSW18} improved the running time bound to $m^{1+o(1)} + n^{1.875+o(1)}\delta^{-1.75}.$

The current approaches to estimating the determinant are based on the construction of determinant sparsifiers. A determinant sparsifier of a graph is a sparse (reweighted) subgraph that approximately preserves the determinant.
\cite{DurfeePPR17} showed how to construct $(1+\delta)$-determinant sparsifiers for general graphs with $\Otil(n^{1.5}\delta^{-1})$ edges.\footnote{Theorem 1.1 in \textcite{DurfeePPR17} states the number of edges required as $O(n^{1.5}\delta^{-2}).$ The same construction with a better choice of parameters leads to determinant-sparsifiers with $O(n^{1.5}\delta^{-1})$ edges.}
Even considering random subgraphs of the complete graph $K_n$, the best construction of a determinant sparsifier requires $\Omega(n^{1.5})$ edges~\cite{Janson94}. 
This results in a natural lower bound of $\Omega(n^{1.5})$ for the running time of current approaches to determinant estimation. That leads to the following question:
\begin{quote}
    Can we compute an $O(1)$-approximation to the Laplacian determinant in $\Otil(n^{1.5})$ time? 
\end{quote}

\paragraph{Estimating Effective Resistances.}
As noted above, the current algorithms for estimating Laplacian determinants are far from this $\Omega(n^{1.5})$ time bound.
The bottleneck in the current approaches is the time required for constructing determinant sparsifiers. 
This, in turn, requires estimating the effective resistances of $s=\Omega(n^{1.5})$ edges within a multiplicative factor of $(1+\eps)$ for $\eps = O(n^{-0.25}).$ The effective resistance between $u,v \in V$ is denoted by $R(u,v)$ or $R^G(u,v)$ and is defined as the equivalent effective resistance between the vertices $u,v$ if the graph is thought of as an electrical network with edge $e$ having resistance $1/w_e.$ We can also write $R(u,v) = (\bm{1_u}-\bm{1_v})^{\top}\LL^+(\bm{1_u}-\bm{1_v})$.
This motivates the following question: 
\begin{quote}
    How efficiently can we compute $(1+\eps)$-approximations to effective resistances of $s$ edges?
\end{quote}

Table~\ref{table:r-eff-algos} summarizes the known algorithms for estimating effective resistances.
The current best algorithm estimates $s$ effective resistances in $m^{1+o(1)} + (n+s)n^{o(1)}\eps^{-1.5}$ time~\cite{ChuGPSSW18}, resulting in an $m^{1+o(1)}+n^{1.875+o(1)}$ time bound for estimating the determinant.
Observe that any $\eps^{-1}$ factors on $s$ will result in a running time worse than $O(n^{1.5})$ for estimating the determinant.

\begin{table}
\begin{tabular}{ c|c|c|c }
Citation & Total running time & Det. Running time & Key ideas \\
\hline
\cite{SpielmanS08} & $m\eps^{-2} + s\eps^{-2}$ & $m\sqrt{n} + n^2$ & JL + Lap. Solvers \\
& $m + n\eps^{-4} + s\eps^{-2}$ & $n^2$ & Spectral-sparsifiers + JL + Lap. Solvers \\
\cite{JambulapatiS18} & $n^2\eps^{-1}$ & $n^{2.25}$ & Spectral-sketches + Lap. Solvers \\
\cite{DurfeeKPRS17} & $m + n\eps^{-2} + s\eps^{-2}$ & $n^2$ & Approx. Schur  \\
\cite{ChuGPSSW18} & $m +  n\eps^{-1.5} + s\eps^{-1.5}$ & $m + n^{1.875}$ & Spectral-sketches + Approx. Schur \\
This paper & $m+n\eps^{-2} + s$ & $m + n^{1.5}$ &  only for expanders, additive approx. to $\bm{\sigma_u}$
\end{tabular}
\caption{\label{table:r-eff-algos}
Known algorithms for approximating effective resistances of $s$ vertex pairs  to within $(1+\varepsilon)$ in a graph with $m$ edges and $n$ vertices, and the implied running time for estimating determinants.
The running times expressed above hide $m^{o(1)}$ factors. JL refers to methods using the Johnson–Lindenstrauss lemma~\cite{JohnsonL84}, and "Lap. Solvers" indicate using nearly-linear time Laplacian Solvers~\cite{SpielmanT04}. "Approx. Schur" indicates using approximate Schur complements~\cite{KyngS16}.}
\end{table}

\paragraph{Data-Structure view.}
Say we're allowed to preprocess $G$ to build a data-structure that occupies small space. After pre-processing, pairs of vertices arrive online and the data-structure must answer a $(1+\eps)$-approximate effective resistance in $G$ for each pair. The trivial approach is to compute and store all pairwise effective resistances using $O(n^2)$ space, and answer each query in $O(1)$ time.

The first non-trivial approach is to use \textcite{SpielmanS08} algorithm for estimating effective resistances combining Johnson-Lindenstrauss dimension reduction with fast Laplacian solvers~\cite{SpielmanT04}, resulting in a data-structure that requires $O(n\eps^{-2} \log n)$ storage\footnote{For convenience, we are counting the complexity in terms of words of $O(\log n)$ size, so the bit complexity is higher by a factor of $O(\log n)$}. The data structure can be built in $\Otil(\min\{m\eps^{-2}, m + n\eps^{-4}\})$ time, and can answer each query efficiently in $O(\eps^{-2} \log n)$ time. 
Another approach is to store a spectral sketch~\cite{JambulapatiS18, ChuGPSSW18}, a sparse subgraph of $G$ that approximately preserves effective resistances. 
Such a sketch can be computed in $\Otil(m)$ time, and requires $\Otil(n\eps^{-1})$
space.
However, the computation time now is much slower, requiring $\Otil(n\eps^{-1})$ time.
There is an immediate open question:
\begin{quote}
    Can we design a data-structure that requires $\widetilde{O}(n\eps^{-1})$ space, and can answer effective resistance queries efficiently, say $\widetilde{O}(\eps^{-1})$  or $O(1)$ time per query? 
\end{quote}
We note that if we can build such a data-structure efficiently, it immediately implies an algorithm for estimating effective resistances. However, building such a data-structure is a harder question than estimating effective resistances.
\subsection{Our Results}
We make progress on all aforementioned questions on expander graphs.
\begin{theorem}
\label{thm:intro-thm1}
For a graph $G$ with expansion $\widetilde{\Omega}(1),$ for any $\eps > 0$, in time $\Otil(m+n\eps^{-2})$, we can compute a data-structure with storage $\Otil(n\eps^{-1}),$ that can answer $(1+\eps)$-approximations to the effective resistance between any vertex pair with high probability in $\Otil(1)$ time.
This gives an algorithm that with high probability returns  $(1+\eps)$-approximate effective resistances for $s$ vertex pairs in $G$ in  $\Otil(m+n\eps^{-2}+s)$ time.
\end{theorem}

This establishes for the first time the existence of a data-structure with $\Otil(n\eps^{-1})$ space that can answer each query very efficiently -- in $\widetilde{O}(1)$ time; the previous best running time with the same amount of space was $\widetilde{\Omega}(n\eps^{-1}).$ 
Note that for $\eps = o(1),$ this is the first data-structure that can answer effective resistance queries in $\Otil(1)$ time while only using  $o(m)$ space, i.e. without essentially storing them all.
For most applications, including graph sparsification and estimating Laplacian determinants, we have $s \ge \widetilde{\Omega}(m)$ or $s \ge \widetilde{\Omega}(n\eps^{-2}),$ resulting in an improvement of $\Omega(\eps^{-1.5})$ in running time compared to the $m^{1+o(1)} + (n+s)n^{o(1)}\eps^{-1.5}$ bound from~\cite{ChuGPSSW18}. 

Building on this, we give an improved algorithm for estimating Laplacian determinants for expander graphs. 
\begin{theorem}
\label{thm:intro-thm2}
For graphs $G$ with $\widetilde{\Omega}(1)$ expansion, for any $\delta > 0,$ there is an algorithm that returns a $(1+\delta)$-approximation of the Laplacian determinant in $\Otil(m+n^{1.5}\delta^{-1})$ with high probability.
\end{theorem}

This improves on the previous best running time of  $m^{1+o(1)}+n^{1.875+o(1)}\delta^{-1.75}$ from~\cite{ChuGPSSW18}. Ours is the first algorithm that matches the natural $\Omega(n^{1.5})$ lower bound on the size of determinant sparsifiers. 
Our algorithm extends to estimating determinants of a symmetric $(1 + \alpha)$-diagonally dominant (DD) matrices (a submatrix $\MM$ of any Laplacian that 
satisfies $\bm{M}_{uu} \ge (1 + \alpha) \sum_{v : v \neq u} |\bm{M}_{uv}|$). 
\begin{theorem}
\label{thm:intro-thm3}
    For any $\delta,\alpha > 0,$ given a symmetric  $(1+\alpha)$-DD matrix $\MM \in \rea^{n \times n},$ with $m$ non-zeroes, there is an algorithm that returns a $(1+\delta)$-approximate determinant for $\MM$ in $\Otil(m + n^{1.5}\delta^{-1}(1+\alpha^{-3}))$ time with high probability.
\end{theorem}
This is significant since the current approach to estimating Laplacian determinants~\cite{DurfeePPR17, ChuGPSSW18} is based on recursing onto two smaller matrices, each with roughly half the vertices, where one of them is  $2$-DD.

\subsection{Technical overview}
Effective resistance is usually studied as an $\ell_2^2$ norm of vectors. Specifically, the approach from \textcite{SpielmanS08} is based on constructing vectors $\bm{z_u}$ for every vertex $u$ such that $R(u,v) = \norm{\bm{z_u} - \bm{z_v}}_2^2$ for every $u,v.$
Applying Johnson-Lindenstrauss dimensionality reduction to these vectors gives vectors $\bm{\widetilde{z}_u}$ in $\Otil(\eps^{-2})$ dimensions such that  $R(u,v) = (1\pm \eps)\norm{\bm{\widetilde{z}_u} - \bm{\widetilde{z}_v}}_2^2.$ It is known that the JL dimension reduction strategy requires $\widetilde{\Omega}(\eps^{-2})$ dimensions~\cite{LarsenN17}.

Spectral sparsifiers~\cite{SpielmanS08, BSS12} and spectral sketches~\cite{JambulapatiS18, ChuGPSSW18} guarantee that there is a sparse graph $H$ such that for any fixed vector $\xx \in \rea^V,$  $\xx^{\top} \LL_H^+ \xx$ is the same as  $\xx^{\top} \LL_G^+ \xx$ up to a $(1\pm\eps)$ factor with high probability. Equivalently, the "error" $ \abs{\xx^{\top} \LL_H^+ \xx -  \xx^{\top} \LL_G^+ \xx}$ is bounded by $\eps \xx^{\top} \LL_G^+ \xx,$ which for expander graphs is  $\Otil({\eps \norm{\xx}_2^2}).$ 
Picking $\xx = \bm{1_u}-\bm{1_v}$ implies that this guarantee is sufficient for approximating effective resistances. 
While the notion of a spectral sketch is stronger than our goal of approximating all effective resistances, it is currently the only known approach to achieve this goal while requiring only $\Otil(n\eps^{-1})$ space.
Furthermore, it is known that every spectral sketch requires $\widetilde{\Omega}(n\eps^{-1})$ space~\cite{andonickqwz16}.

We observe that for estimating effective resistances on expanders, a much weaker error guarantee of $\eps \norm{\xx}_1^2$ suffices. This is because the $\ell_1$ norm and $\ell_2$ norm of $ \bm{1_u}-\bm{1_v}$ are within a constant factor, though for arbitrary vectors $\xx,$ $\norm{\xx}^2_1$ could be larger than $\norm{\xx}_2^2$ by a factor of $n.$
%
Building on this observation, we prove that for any graph,  $\eps$-additive approximations to $\{\bm{D}\bm{L}^+\PPi\bm{1_u}\}$ for all vertices $u,$ are sufficient for recovering $(1 + O(\eps))$-approximations to effective resistances between any vertex pair, where $\PPi$ is the projection orthogonal to the all ones vector.

However, the vectors $\{\bm{D}\bm{L}^+\PPi\bm{1_u}\}$ do not have small enough $\ell_1$ norm for us to obtain a sparsity guarantee with $\epsilon$-additive approximations, even on expanders. We instead show that there exist vectors $\{\bm{\sigma_u}\}$ which are similar to  $\{\bm{D}\bm{L}^+\PPi\bm{1_u}\}$ for which the previous argument still holds. Furthermore, in an expander graph, for every vertex $u$, the vector $\bm{\sigma_u}$ can only have $\Otil(\eps^{-1})$ coordinates larger than $\eps.$
The proof is based on interpreting these vectors in terms of random walks starting at the vertex $u,$ and exploiting that random walks in expanders mix in $\Otil(1)$ steps.
Thus, we only need $\Otil(n\eps^{-1})$ space to store $\eps$-additive approximations to all vectors $\{\bm{\sigma_u}\}_{u \in V},$ and can return $(1+\eps)$-approximations to each effective resistance query in $\Otil(1)$ time.

This proof also gives an algorithm for building such a data-structure.  Starting from each vertex $u,$ we will take $\Otil(\eps^{-2})$ random walks, each of length $\Otil(1).$ Standard concentration bounds guarantee that this is sufficient to estimate each coordinate of the vector up to an additive $\eps$ error with high probability. We preprocess the graph in $O(m)$ time to be able to sample each random walk step in $O(1)$ time. This results in a total preprocessing time of $\Otil(m+n\eps^{-2})$ and an algorithm for computing $s$ pairs of effective resistances in total time   $\Otil(m+n\eps^{-2} +s).$ 

In order to $(1+\delta)$-approximate the Laplacian determinant for expander graphs and $(1+\alpha)$-DD matrices, we build on the algorithm from~\cite{DurfeePPR17}.
We modify their algorithm to ensure that every subgraph generated by the algorithm, for which we need to estimate its Laplacian determinant, remains an expander.
The runtime bottleneck in our algorithm is the same as in~\cite{DurfeePPR17}, the time for estimating the effective resistance of $s = \Otil(n^{1.5}\delta^{-1})$ edges up to a $(1+\Theta(n^{-0.25}\delta^{-0.5}))$ multiplicative factor. Thus, with our modification ensuring expander graphs throughout, we can use our new algorithm for estimating effective resistances,
resulting in a running time of $\Otil(m+n^{1.5}\delta^{-1}).$

For a $(1+\alpha)$-DD matrix $\MM$, by adding a new vertex and connecting it to all existing vertices,  this matrix can be extended to a Laplacian. Crucially, the pseudodeterminant of the new matrix is exactly $n$ times the determinant of $M.$ If we eliminate this new vertex in the Schur complement, we obtain a graph with expansion $\alpha/(1+\alpha)$ (as measured by the second eigenvalue of its normalized Laplacian). However, this new graph is dense and hence cannot be written down explicitly.
Instead, we exploit the fact that we can simulate random walks on this graph by representing the clique implicitly. This allows us to implement our algorithm for estimating effective resistances for expanders on the Schur complement, and hence estimate its determinant in $\Otil(m + n^{1.5}\delta^{-1}(1 + \alpha^{-3}))$ time. 

\section{Preliminaries}

We are given a graph $G(V, E)$ with $n$ vertices and $m$ edges and edge weights bounded by $W$.  Let $\bm{A}$ denote the weighted adjacency matrix for $G,$ and $\bm{D}$ denote the diagonal matrix of weighted degrees. Then, the Laplacian $\bm{L}$ of $G$ is a $V$ by $V$ matrix defined as 
$\bm{L} = \bm{D} - \bm{A}$. 
The normalised Laplacian of $G$ is $\bm{D}^{-1/2}\bm{L}\bm{D}^{-1/2}$. The measure of graph expansion we use is the spectral gap of the normalized Laplacian, or equivalently its second smallest eigenvalue $\nu_2 = \lambda_2(\bm{D}^{-1/2}\bm{L}\bm{D}^{-1/2})$. We say that $G$ is an expander if $\nu_2 = \widetilde{\Omega}(1).$ By Cheeger's inequality, this is equivalent to G having  $\widetilde{\Omega}(1)$ conductance.

Given a symmetric matrix $\bm{M} \in \rea^{n \times n},$ we will express its spectral decomposition as
$
    \bm{M} = \sum_{i=0}^{n}\lambda_i \bm{\psi_i\psi_i}^\top,
$
where $\{\ppsi_i\}_{i=1,\ldots,n}$ are orthonormal eigenvectors with eigenvalues $\lambda_1 \le \lambda_2 \le \ldots \le \lambda_n$ respectively. We define its pseudoinverse $\bm{M}^+$ to be:
$\bm{M}^+ = \sum_{i: \lambda_i \neq 0}\frac{1}{\lambda_i} \bm{\psi_i\psi_i}^\top.$

In particular, for a connected graph $G$, its Laplacian $\bm{L}$ has kernel spanned exactly by $\bm{1}$, the all $1$ vector. 
We will write $\bm{A} \preceq \bm{B}$ or $\bm{B} \succeq \bm{A}$ iff $\bm{B} - \bm{A}$ is positive semi-definite. 

The effective resistance between two vertices $u$ and $v$ in a graph $G$ is defined to be 
\begin{align*}
    R(u,v) = (\bm{1_u} - \bm{1_v})^\top\bm{L}^+(\bm{1_u} - \bm{1_v}).
\end{align*}
where $\bm{1_u}$ is defined to be the vector that is 1 on the $u$ entry, and 0 everywhere else. 
We will  use the following fact about effective resistances, the proof of which is available in the appendix.
\begin{restatable}[]{fact}{factrefflowerbound}
\label{fact:reff_lowerbound}
    The effective resistance of any edge $(u,v)$ is lower bounded by $\frac{1}{2}(\frac{1}{d_u}+\frac{1}{d_v})$, where $d_u$ and $d_v$ are the degrees of $u$ and $v$ respectively. 
\end{restatable}

The Schur complement of a symmetric matrix $\bm{M} = \left[ 
\begin{array}{c c} 
  \bm{F} & \bm{B} \\ 
  \bm{B}^\top & \bm{C}
\end{array} 
\right]$ onto some submatrix $\bm{C}$ is defined to be $\bm{Sc}(\bm{M},\bm{C}) = \bm{C} - \bm{B}^\top\bm{F}^{-1}\bm{B}$. In the context of Laplacians, we abuse notation to have $\bm{Sc}(G,V_1)$, $\bm{Sc}(\bm{L},V_1)$, and $\bm{Sc}(\bm{L},\bm{C})$ all refer to the Schur complement of $\bm{L}$ onto the submatrix $\bm{C}$ on the vertex support $V_1$. 
A useful fact about Schur complements is that the Schur complement of a Laplacian is also a Laplacian. When convenient, we may use $Sc(G,V_1)$ to refer to the graph of the Schur complement from $G$ onto $V_1$. 
Another important fact is that  $Sc(G,V_1)$ preserves the effective resistances of the original graph for pairs of vertices in $V_1.$
\begin{fact}
Consider a graph $G$, and some arbitrary subset of vertices $V_1$. For all $u,v \in V_1$, we have that
  $  R^G(u,v) = R^{Sc(G,V_1)}(u,v).$
\end{fact}

A $(1 + \alpha)$-diagonally dominant matrix $\bm{M}$ is a symmetric matrix $M$ with non-positive off diagonal entries satisfying $\bm{M}_{uu} \ge (1 + \alpha) \sum_{v : v \neq u} |\bm{M}_{uv}|$. Given such a matrix $\bm{M}$ of size $n$, we can complete it into a Laplacian $\bm{L_M}$ by adding only a single column and row. Call the newly added vertex $x$. We have that:
\begin{align*}
    (\bm{L_M})_{ux}= (\bm{L_M})_{xu} &= -\bm{M}_{uu} + \sum_{v:v\neq u} \bm{M}_{uv},\\
    (\bm{L_M})_{xx} &= -\sum_{u}(\bm{L_M})_{xu}.
\end{align*}

We also refer to a subset $V_1$ of a graph $G$ as being $(1 + \alpha)$-diagonally dominant if for each $u \in V_1$, we have:
\begin{align*}
    \sum_{v \sim u, v \notin V_1} w_{uv} \geq \alpha \sum_{v \sim u, v \in V_1} w_{uv}.
\end{align*}

\section{Effective Resistances}

In this section, we demonstrate that for expanders, a sparse data structure exists, and can be efficiently computed, from which we can query the effective resistances between any two vertices in $\widetilde{O}(1)$ time. The main result of this section is a more precise statement of Theorem~\ref{thm:intro-thm1} as follows:
\begin{theorem}\label{thm:main_reff}
    Given a graph $G$ such that the spectral gap of its normalised Laplacian is $\nu_2$, there exists an algorithm \textsc{EffectiveResistanceSketch} that computes vectors $\{\bm{\widetilde{\sigma}_u}\}$, in $O(m + n\varepsilon^{-2}\nu_2^{-3}\log{n}\log^3{\frac{nW}{\varepsilon}})$ total time, such that:
\begin{enumerate}
    \item The vectors $\{\bm{\widetilde{\sigma}_u}\}$ are $O(\nu_2^{-1}\varepsilon^{-1}\log{nW})$-sparse.
    \item For any vertex pair $(u,v)$ \textsc{EffectiveResistanceQuery} can return an estimate of $R(u,v)$ in $\Otil(1)$ time\footnote{The query time is equal to the time required to lookup a constant number of entries. A constant amortized time can be achieved with a good hash function, or $O(\log{n})$ time can be achieved through a sorted data set.} by querying the vectors $\bm{\widetilde{\sigma}_u}$ and $\bm{\widetilde{\sigma}_v}$. 
    This estimate is a $(1+\eps)$-approximation to $R(u,v)$ with high probability.
\end{enumerate}
\end{theorem}

\subsection{Effective Resistances on Expanders}

We first begin by proving the existence of these vectors:
\begin{lemma}\label{lem:main_existence}
   Suppose $G$ is a graph such that the spectral gap of its normalised Laplacian is $\nu_2$. There exists vectors $\{\bm{\widetilde{\sigma}_u}\}$ such that the following hold:
\begin{enumerate}
    \item The vectors $\{\bm{\widetilde{\sigma}_u}\}$ are $O(\nu_2^{-1}\varepsilon^{-1}\log{nW})$-sparse.
    \item The effective resistance of an edge $(u,v)$ can be $(1 + \varepsilon)$-approximated by the vectors $\bm{\widetilde{\sigma}_u}$ and $\bm{\widetilde{\sigma}_v}$ in $\widetilde{O}(1)$ time. 
\end{enumerate}
\end{lemma}

For some intuition as to why we consider the vectors that we do, we first consider the set of vectors $\bm{D}\bm{L}^+ \bm{1_u}, u \in V$. It is easy to see that we can query $(1 + \varepsilon)$-approximations to the effective resistances from $\varepsilon$-additive approximations to these vectors, given the lower bound in Fact~\ref{fact:reff_lowerbound}.

If each of these vectors were to have small $\ell_1$ norm, on the order of $\Otil(\nu_2^{-1})$, then we would be done, since we could just round off these vectors to the nearest $\varepsilon$ to obtain the sparse vectors.  However, it is not in general true for expanders that these vectors have small enough $\ell_1$ norm, and in fact, for constant expanders, the $\ell_1$ norm of $\bm{DL}^+\bm{1_u}$ can scale linearly with $n$. 

Instead, we consider another set of vectors, constructed through writing $\bm{L}^+$ as a power series expansion. We first begin with the power series expansion on $\bm{L}^+$:
\begin{lemma}\label{lem:power_expansion}
    Given a graph $G$, with Laplacian $\bm{L} = \bm{D} - \bm{A}$, projection matrix $\bm{\Pi} = \bm{I} - \frac{1}{n}\bm{11^\top}$ and $\bm{\pi} = \frac{\bm{D1}}{\bm{1}^\top\bm{D1}}$, the stationary distribution of the random walk on $G$, we have:
\begin{align*}
\bm{L}^+\left(\bm{1_u} - \bm{\pi}\right) &= \frac{1}{2}\bm{\Pi D}^{-1}\sum_{t = 0}^\infty \left((\frac{1}{2}\bm{I} + \frac{1}{2}\bm{AD}^{-1})^t \bm{1_u} - \bm{\pi}\right).
\end{align*}
\end{lemma}

\begin{proof}
Let $\bm{X} = \frac{1}{2}\bm{I} + \frac{1}{2}\bm{AD}^{-1}$. We have:
\begin{align*}
    \frac{1}{2}(\bm{D} - \bm{A})\bm{D}^{-1}\sum_{t = 0}^n (\frac{1}{2}\bm{I} + \frac{1}{2}\bm{AD}^{-1})^t &= (\bm{I} - \bm{X})\sum_{t = 0}^n\bm{X}^t
    = \bm{I} - \bm{X}^{n+1}.
\end{align*}

Consider the action of this matrix on $\bm{1_u} - \bm{\pi}$. As $n \rightarrow \infty$, since $\bm{X}$ is the lazy random walk matrix, repeated calls of it on $\bm{1_u} - \bm{\pi}$, a vector with entrywise sum 0, sends it to $\bm{0}$. As such, we have:
\begin{align*}
    \frac{1}{2}(\bm{D} - \bm{A})\bm{D}^{-1}\sum_{t = 0}^\infty (\frac{1}{2}\bm{I} + \frac{1}{2}\bm{AD}^{-1})^t(\bm{1_u} - \bm{\pi}) &= (\bm{I} - \bm{X})\sum_{t = 0}^\infty\bm{X}^t(\bm{1_u} - \bm{\pi})
    = \bm{1_u} - \bm{\pi}.
\end{align*}
Thus, $\bm{x} = \bm{D}^{-1}\sum_{t = 0}^\infty (\frac{1}{2}\bm{I} + \frac{1}{2}\bm{AD}^{-1})^t(\bm{1_u} - \bm{\pi})$ is a solution to the equation $(\frac{1}{2}\bm{L})\bm{x} = \bm{1_u} - \bm{\pi}$. We know that $(\frac{1}{2}\bm{L})^+(\bm{1_u} - \bm{\pi})$ is the unique solution to this equation that is perpendicular to $\bm{1}$, so we have:
\begin{align*}
    2\bm{L}^+(\bm{1_u} - \bm{\pi}) = (\frac{1}{2}\bm{L})^+(\bm{1_u} - \bm{\pi})
    &= \bm{\Pi D}^{-1}\sum_{t = 0}^\infty (\frac{1}{2}\bm{I} + \frac{1}{2}\bm{AD}^{-1})^t(\bm{1_u} - \bm{\pi})\\
    \bm{L}^+(\bm{1_u} - \bm{\pi}) &= \frac{1}{2}\bm{\Pi D}^{-1}\sum_{t = 0}^\infty \left((\frac{1}{2}\bm{I} + \frac{1}{2}\bm{AD}^{-1})^t \bm{1_u} - \bm{\pi}\right).
\end{align*}
as desired. 
\end{proof}

Instead of considering $\bm{DL}^+\bm{1_u}$, we instead consider $\bm{\sigma_u} = \frac{1}{2}\sum_{t = 0}^\infty \left((\frac{1}{2}\bm{I} + \frac{1}{2}\bm{AD}^{-1})^t \bm{1_u} - \bm{\pi}\right)$. This is similar in spirit to $\bm{DL}^+\bm{1_u}$, since $\bm{DL}^+(\bm{1_u} - \bm{\pi}) = \bm{D\Pi D}^{-1} \bm{\sigma_u}$. 

However, compared to $\bm{DL}^+\bm{1_u}$, it can be shown that $\bm{\sigma_u}$ each have small $\ell_1$ norm, and can be used to recover $(1 + \varepsilon)$-approximations to the effective resistances. We first begin with a by showing that $\varepsilon$-additive approximations to this set of vectors suffices. 
\begin{lemma}\label{lem:sigma_recover_approx}
    Given an $\frac{1}{4}\varepsilon$-additive approximation to the vectors $\{\bm{\sigma_u} = \frac{1}{2}\sum_{t = 0}^\infty \left((\frac{1}{2}\bm{I} + \frac{1}{2}\bm{AD}^{-1})^t \bm{1_u} - \bm{\pi}\right) | u \in V\}$, we can obtain $(1 + \varepsilon)$-approximations to the effective resistances of the graph $G$ corresponding to the Laplacian $\bm{L}$. 
\end{lemma}

\begin{proof}
Suppose we have $\varepsilon$-additive approximations $\{\bm{\widetilde{\sigma}_u}\}$ to $\{\bm{\sigma_u}\}$. For a given edge $(u,v)$, we have that:
\begin{align*}
     R(u,v) &= (\bm{1_u} - \bm{1_v})^\top\bm{L}^+(\bm{1_u} - \bm{1_v})\\
    &= (\bm{L}^+\bm{1_u})_u - (\bm{L}^+\bm{1_v})_u - (\bm{L}^+\bm{1_u})_v + (\bm{L}^+\bm{1_v})_v.
\end{align*}

Since $\bm{L}^+(\bm{1_u} - \bm{\pi}) = \bm{\Pi D}^{-1}\bm{\sigma_u}$ from Lemma~\ref{lem:power_expansion}
and since the projection to the all $1$s vector does not change the difference between any two elements, we have:
\begin{align*}
    (\bm{L}^+(\bm{1_u} - \bm{\pi}))_u - (\bm{L}^+(\bm{1_u} - \bm{\pi}))_v = (\bm{\Pi D}^{-1}\bm{\sigma_u})_u - (\bm{\Pi D}^{-1}\bm{\sigma_u})_v = (\bm{D}^{-1}\bm{\sigma_u})_u - (\bm{D}^{-1}\bm{\sigma_u})_v.
\end{align*}
Next, since they are both offset by the same value $(\bm{L}^+\bm{\pi})_u$, we have:
\begin{align*}
    (\bm{L}^+\bm{1_u})_u - (\bm{L}^+\bm{1_v})_u = (\bm{L}^+(\bm{1_u} - \bm{\pi}))_u - (\bm{L}^+(\bm{1_v} - \bm{\pi}))_u,
\end{align*}
which gives us:
\begin{align*}
    R(u,v) &= (\bm{L}^+\bm{1_u})_u - (\bm{L}^+\bm{1_u})_v - (\bm{L}^+\bm{1_v})_u + (\bm{L}^+\bm{1_v})_v\\
    &= (\bm{L}^+(\bm{1_u} - \bm{\pi}))_u - (\bm{L}^+(\bm{1_u} - \bm{\pi}))_v - (\bm{L}^+(\bm{1_v} - \bm{\pi}))_u + (\bm{L}^+(\bm{1_v} - \bm{\pi}))_v\\
    &= (\bm{D}^{-1}\bm{\sigma_u})_u - (\bm{D}^{-1}\bm{\sigma_u})_v - (\bm{D}^{-1}\bm{\sigma_v})_u + (\bm{D}^{-1}\bm{\sigma_v})_v.
\end{align*}
Hence, we can approximate the effective resistance with:
\begin{align*}
     \frac{1}{d_u}(\bm{\widetilde{\sigma}_u})_u - \frac{1}{d_v}(\bm{\widetilde{\sigma}_u})_v + \frac{1}{d_v}(\bm{\widetilde{\sigma}_v})_v - \frac{1}{d_u}(\bm{\widetilde{\sigma}_v})_u.
\end{align*}
The total error is upper bounded by $\frac{1}{2}(\frac{1}{d_u} + \frac{1}{d_v})\varepsilon$. Since by Fact~\ref{fact:reff_lowerbound}, the effective resistance $R(u,v)$ is lower bounded by $\frac{1}{2}(\frac{1}{d_u}+\frac{1}{d_v})$, this gives us a $(1 + \varepsilon)$-approximation to the effective resistances.
\end{proof}

Next, we show that these vectors do indeed have small $\ell_1$ norm. 
\begin{lemma}\label{lem:sigma_sparse}
    Given a graph $G$ so that the spectral gap of its normalised Laplacian is $\nu_2$, the vectors $\{\bm{\sigma_u} = \frac{1}{2}\sum_{t = 0}^\infty \left((\frac{1}{2}\bm{I} + \frac{1}{2}\bm{AD}^{-1})^t \bm{1_u} - \bm{\pi}\right) | u \in V\}$ each have $\ell_1$ norm bounded by $O(\nu_2^{-1}\log{nW})$.
\end{lemma}

\begin{proof}
First, notice that this summation converges. $\frac{1}{2}\bm{I} + \frac{1}{2}\bm{AD}^{-1}$ is the lazy random walk matrix. As such, when applied to $\bm{1_u}$, a vector with sum 1, this term goes to $\bm{\pi}$, the stationary distribution. We explicitly bound the tail terms of this summation, truncating it at $t_0$. We have:
\begin{align*}
    \sum_{t = 0}^\infty \left((\frac{1}{2}\bm{I} + \frac{1}{2}\bm{AD}^{-1})^t \bm{1_u} - \bm{\pi}\right) &= \sum_{t = 0}^{t_0-1} \left((\frac{1}{2}\bm{I} + \frac{1}{2}\bm{AD}^{-1})^t \bm{1_u} - \bm{\pi}\right) + \sum_{t = t_0}^\infty \left((\frac{1}{2}\bm{I} + \frac{1}{2}\bm{AD}^{-1})^t \bm{1_u} - \bm{\pi}\right).
\end{align*}
We first bound the total $\ell_1$ norm of the second term. We begin with a well known result on the convergence of random walks, the proof of which is available in the appendix.
\begin{restatable}[]{lemma}{lemrandomwalk}
\label{lem:randomwalk}
    Given a graph $G$, and an initial distribution $\bm{p}$, the $\ell_1$ norm difference between the distribution of the lazy random walk after $t$ steps and the stationary distribution is bounded by:
    \begin{align*}
        \norm{(\frac{1}{2}\bm{I} + \frac{1}{2}\bm{AD}^{-1})^{t}\bm{p} - \bm{\pi}}_1 \le e^{-t\nu_2/2} \frac{nd_{\max}}{d_{\min}}.
    \end{align*}
where $d_{max}$ and $d_{min}$ are the largest and smallest weighted degrees respectively. 
\end{restatable}

This bounds the $\ell_1$ norm of the second term:
\begin{align*}
    \norm{\sum_{t = t_0}^\infty \left((\frac{1}{2}\bm{I} + \frac{1}{2}\bm{AD}^{-1})^t \bm{1_u} - \bm{\pi}\right)}_1 &\leq \sum_{t = t_0}^\infty \norm{\left((\frac{1}{2}\bm{I} + \frac{1}{2}\bm{AD}^{-1})^t \bm{1_u} - \bm{\pi}\right)}_1\\
    &\leq \frac{2nd_{max}}{d_{min}} \sum_{t = t_0}^\infty e^{-t\nu_2/2} \\ 
    &=  \frac{2nd_{max}}{d_{min}} \frac{e^{-\frac{\nu_2 t_0}{2}}}{1 - e^{-\frac{\nu_2}{2}}}.
\end{align*}

We have that $\frac{d_{max}}{d_{min}} \leq W$, and $(1 - e^{-\frac{\nu_2}{2}}) \geq \nu_2/4$. We bound $\nu_2$ as follows. The conductance of any graph is minimally $\frac{1}{n^2W}$. Cheeger's inequality then tells us that $\nu_2 \geq (\frac{1}{n^2W})^2/2$. Hence, setting $t_0 = O(\nu_2^{-1}\log{nW})$ bounds the error by a constant. The first term can also easily be bounded by $2t_0$:
\begin{align*}
\norm{\sum_{t = 0}^{t_0-1} \left((\frac{1}{2}\bm{I} + \frac{1}{2}\bm{AD}^{-1})^t \bm{1_u} - \bm{\pi}\right)}_1 
&\leq \sum_{t = 0}^{t_0-1} \norm{(\frac{1}{2}\bm{I} + \frac{1}{2}\bm{AD}^{-1})^t\bm{1_u}}_1 + \norm{\bm{\pi}}_1\\
&= 2t_0,
\end{align*}
since the random walk matrix preserves $\ell_1$ norms for strictly positive vectors. Combining bounds the $\ell_1$ norm of $\sum_{t = 0}^\infty \left( (\frac{1}{2}\bm{I} + \frac{1}{2}\bm{AD}^{-1})^t \bm{1_u} - \bm{\pi}\right)$ by $O(\nu_2^{-1}\log{nW})$ as desired. 
\end{proof}

We now combine Lemma~\ref{lem:sigma_recover_approx} and Lemma~\ref{lem:sigma_sparse} to prove Lemma~\ref{lem:main_existence}:

\begin{proof}[Proof of Lemma~\ref{lem:main_existence}]

We have from Lemma~\ref{lem:sigma_sparse} that the $\ell_1$ norm of $\bm{\sigma_u} = \sum_{t = t_0}^\infty \left( (\frac{1}{2}\bm{I} + \frac{1}{2}\bm{AD}^{-1})^t \bm{1_u} - \bm{\pi}\right)$ is bounded by $O(\nu_2^{-1}\log{nW})$. As such, there can be at most $O(\nu_2^{-1}\varepsilon^{-1}\log{nW})$ entries in this vector that have absolute value greater than $\varepsilon/4$. Consider the $\varepsilon/4$-additive approximation to this vector $\bm{\widetilde{\sigma}_u}$ where each entry that has absolute value less than $\varepsilon/4$ is 0. This vector is $O(\nu^{-1}_2\varepsilon^{-1}\log{nW})$ sparse. 

By Lemma~\ref{lem:sigma_recover_approx}, we can use these $\varepsilon/4$-additive approximations to the vectors $\bm{\sigma_u}$ to obtain $(1 + \varepsilon)$-approximations to the effective resistances in $\widetilde{O}(1)$ time for each query. 
\end{proof}

Finally, we demonstrate that these vectors can be efficiently calculated, proving Theorem~\ref{thm:main_reff}. From the above discussion, we've established that to obtain $(1 + \epsilon)$-approximations to the effective resistances, we simply have to produce the vectors $\{\bm{\widetilde{\sigma}_u}\}_{u \in V}$. We can view the algorithm as follows. The random walk starting at any position mixes quickly, and quickly there is no significant difference between its probability distribution and the stationary distribution. Most coordinates approach the stationary distribution quickly enough that even summing across the first $\Otil(\nu_2^{-1})$ steps their contributions are small and can be discarded, while the remaining $\Otil(\epsilon^{-1})$ coordinates have significant sums from which we can extract the effective resistances. 

\begin{proof}[Proof of Theorem~\ref{thm:main_reff}]
From our proof of Lemma~\ref{lem:sigma_sparse}, we have that the power series, truncated at $t_0 = O(\nu^{-1}_2\log{\frac{nW}{\varepsilon}})$ gives us at most $\varepsilon/8$ additive error in each coordinate. As such, we only have to produce a $\varepsilon/8$-approximation of:
\begin{align*}
    \sum_{t = 0}^{t_0-1} \left( (\frac{1}{2}\bm{I} + \frac{1}{2}\bm{AD}^{-1})^t \bm{1_u} - \bm{\pi}\right).
\end{align*}
To do so, we simply perform lazy random walks. We calculate an $\varepsilon/8$-additive approximation to $\left(\sum_{t = 0}^{t_0 - 1} (\frac{1}{2}\bm{I} + \frac{1}{2}\bm{AD}^{-1})^t \bm{1_u}\right)_v$ for all $v \in V$ by performing $s = O(\varepsilon^{-2}t_0\log{n})$ random walks for each length $l$ from $0$ to $t_0$ starting at $u$. Fix some starting vertex $u$, and some end vertex $v$, and let $X^{uv}_{li}$ be the random variable that the $i$-th random walk of length $l$ starting at $u$ ends at $v$. Let $S_{uv} = \sum_{i = 1}^s\sum_{l = 0}^{t_0 - 1}X^{uv}_{li}$. 

\begin{algorithm} \label{alg:det}
\caption{\textsc{EffectiveResistanceSketch}($G,\varepsilon,\nu_2$)}\label{alg:expandereffectiveresistance}
\begin{algorithmic}
\State \textbf{Input:} Graph $G$, with spectral gap of its normalised Laplacian being $\nu_2$
\State \textbf{Output} $O(\nu_2^{-1}\epsilon^{-1}\log{nW})$-sparse vectors $\{\bm{\sigma_u}\}$ from which the effective resistances can be queried in $O(1)$ time. 
\State $\bm{S_{\frac{1}{n}1}}$ $\gets$ $\bm{0}$
\For{$u \in V$}
    \State $\bm{S_u}$ $\gets$ $\bm{0}$
    \For{$i = 1\;to\;s = O(\varepsilon^{-2}\nu^{-1}_2\log{n}\log{\frac{nW}{\varepsilon}})$}
        \For{$l = 0\;to\;t_0 - 1 = O(\nu^{-1}_2\log{\frac{nW}{\varepsilon}})$}
            \State Perform a length $l$ random walk starting from $u$, ending at vertex $v$.
            \State $(\bm{S_u})_v$ $\gets$ $(\bm{S_u})_v + \frac{1}{s}$
        \EndFor
    \EndFor
\EndFor
\For{$u \in V$}
    \State $\bm{\widetilde{\sigma}_u}$ $\gets$ $\bm{S_u} - t_0\bm{\pi}_u$
    \State For each entry of $\bm{\widetilde{\sigma}_u}$, if it smaller than $\varepsilon$ in absolute value, set it to $0$. 
    \EndFor
    \State \Return $\{\bm{\widetilde{\sigma}_u}\}_{u \in V}$ as sparse vectors
\end{algorithmic}
\end{algorithm}

\begin{algorithm} \label{alg:query}
\caption{\textsc{EffectiveResistanceQuery}($\{\bm{\widetilde{\sigma}_u}\},(x,y)$)}\label{alg:effectiveresistancequery}
\begin{algorithmic}
\State \textbf{Input:} The output of \textsc{EffectiveResistanceSketch}, $\{\bm{\widetilde{\sigma}_u}\}$, and two vertices to be queried $(x,y)$
\State \textbf{Output} $(1 + \varepsilon)$-approximation of the effective resistance between $x$ and $y$
\State \Return $\frac{1}{d_u}(\bm{\widetilde{\sigma}_u})_u - \frac{1}{d_v}(\bm{\widetilde{\sigma}_u})_v + \frac{1}{d_v}(\bm{\widetilde{\sigma}_v})_v - \frac{1}{d_u}(\bm{\widetilde{\sigma}_v})_u$.
\end{algorithmic}
\end{algorithm}

By a Hoeffding bound, we have that:
\begin{align*}
\Pr(|S_{uv} - \mathbb{E}[S_{uv}]| \geq  s\varepsilon/8 ) \leq 2e^{-\frac{2(s\varepsilon/8)^2}{st_0}}.
\end{align*}
This gives us that our approximation to $\left(\sum_{t = 0}^{t_0 - 1} (\frac{1}{2}\bm{I} + \frac{1}{2}\bm{AD}^{-1})^t \bm{1_u}\right)_v$, namely $\frac{1}{s}S_{uv}$ is an $\varepsilon/8$-additive approximation with high probability. Now union bounding over all the possible vertices gives us high probability guarantees on all the errors being $\varepsilon/8$-additive approximations. 

For each vertex, there are a total of $s = O(\varepsilon^{-2}t_0\log{n})$ random walks being performed for each length $l$ from $0$ to $t_0$. We note here that there is a procedure \textsc{UnsortedProportionalSampling}\cite{Walker77,BringmannP16} with $O(m)$ total preprocessing time, from which we can query random walk edges in $O(1)$ time each. A simpler approach using a balanced binary search tree would also suffice, but with $O(\log{n})$ overhead. This efficient sampling method allows us to run these random walks in total time:
\begin{align*}
    O(m + nst^2_0) = O(m + n\varepsilon^{-2}\nu_2^{-3}\log{n}\log^3{\frac{nW}{\varepsilon}}).
\end{align*}

Now, subtracting $t_0\bm{\pi}$ from the vector yields a $\varepsilon/8$-additive approximation to $\sum_{t = 0}^{t_0 - 1} \left( (\frac{1}{2}\bm{I} + \frac{1}{2}\bm{AD}^{-1})^t \bm{1_u} - \bm{\pi}\right))$ as desired.

The full algorithm to sketch these vectors can be seen in Algorithm~\ref{alg:expandereffectiveresistance}, and the query algorithm in Algorithm~\ref{alg:effectiveresistancequery}.
\end{proof}

\subsection{Effective Resistances on $(1+\alpha)$-DD Matrices}

We prove an analogous statement for $(1+\alpha)$-DD Matrices.

Given a $2$-DD matrix $\bm{M}$, we look at its completion into a Laplacian $\bm{L_M}$ representing graph $G_M$ with new added vertex $x$. On this new graph $G_M$, we show that the effect resistances of every edge not involving $x$ can be calculated.
\begin{lemma}\label{lem:2ddreff}
   Suppose $\bm{M}$ is a $(1 + \alpha)$-DD matrix, and let $\bm{L_M}$ represent its completion into a Laplacian, and $G_M$ be the represented graph with new vertex $x$. There is an algorithm that builds a data-structure in $O(m + n\varepsilon^{-2}(1 + \alpha^{-3})\log{n}\log^3{\frac{nW}{\varepsilon}})$ time, where where $m$ is the number of edges in $G_M$. 
   This data-structure allows us to query for to the effective resistances between any two original vertices  in $G_M$ (not x) in $\widetilde{O}(1)$ time, and with high probability, it returns 
    a $(1 + \varepsilon)$-approximation.
\end{lemma}

To calculate the effective resistances between vertices in $\bm{M}$, we can calculate the effective resistances in the Schur complement of $G_M$ onto $V\backslash\{x\}$. Since we are removing a single vertex $x$, we are adding a weighted clique back onto the graph. We first show that the Schur complement produced is a $\nu_2$ expander. 

Let $G_{M}{[V\backslash x]}$ denote the subgraph of $G_M$ restricted to the vertices $V\backslash \{x\}$. Let the degrees of each vertex $u$ in $G_{M}{[V\backslash x]}$ be $d_u$, and the weight of the edge $(u,x)$ be $d'_u$. Since $\bm{M}$ is an $(1+\alpha)$-DD matrix, we have that $d'_u \geq \alpha d_u$. Let $d_x = \sum_u d'_u$ be the total degree of the new vertex $x$. It is well known\footnote{A proof is available in the preliminaries section of \cite{KyngS16}} that the weighted clique added by the Schur complement has weights $\frac{d'_ud'_v}{d_x}$ for the edge $(u,v)$. 

Let $\bm{D'}$ be the diagonal matrix with $\bm{D'}_{uu} = d'_u$. Consider the Laplacian on just the weighted clique, normalised by the degrees $\bm{D'}$, $\bm{D'}^{-1/2}(\bm{Sc}(G_M,V\backslash\{x\}) - \bm{L_{G_M[V\backslash \{x\}]}})\bm{D'}^{-1/2}$. We have that if $\lambda$ is an eigenvalue of this matrix, with eigenvector $\bm{\psi}$, that for each $u$:
\begin{align*}
    \frac{1}{\sqrt{d'_u}}d'_u\frac{1}{\sqrt{d'_u}}\bm{\psi}_u - \frac{1}{\sqrt{d'_u}}\sum_v \frac{d'_ud'_v}{d_x}\frac{1}{\sqrt{d'_v}}\bm{\psi}_v = \lambda \bm{\psi}_u,\\
    -\sum_v \frac{d'_v}{d_x}\frac{1}{\sqrt{d'_v}}\bm{\psi}_v = (\lambda - 1)\frac{1}{\sqrt{d'_u}} \bm{\psi}_u.
\end{align*}

In particular, we have that $(\lambda - 1)\frac{1}{\sqrt{d'_u}} \bm{\psi}_u = (\lambda - 1)\frac{1}{\sqrt{d'_v}} \bm{\psi}_v$ for any two $u,v$. This gives us an eigenvector $\bm{D'}^{1/2}\bm{1}$, with eigenvalue $0$, and any other vectors perpendicular to $\bm{D'}^{1/2}\bm{1}$ having eigenvalue $1$. As such, we have:
\begin{align*}
    \bm{D'}^{-1/2}(\bm{Sc}(G_M,V\backslash\{x\}) - \bm{L_{G_M[V\backslash \{x\}]}})\bm{D'}^{-1/2} \succeq \bm{I} - \frac{1}{\bm{1}^\top\bm{D'}\bm{1}}\bm{D'}^{1/2}\bm{1}(\bm{D'}^{1/2}\bm{1})^\top.
\end{align*}
Let $\bm{D}$ be the diagonal with the degrees in $\bm{Sc}(G_M,V\backslash\{x\})$, $\bm{D}_u = d_u + d'_u - \frac{(d'_u)^2}{d_x}$. Since $d'_u \geq \alpha d_u$, we have that each diagonal entry of $\bm{D'}\bm{D}^{-1}$ is lower bounded by $d'_u/(d_u + d'_u) \ge \alpha/(1 + \alpha)$. We now have:
\begin{align*}
    (\bm{Sc}(G_M,V\backslash\{x\}) - \bm{L_{G_M[V\backslash \{x\}]}}) &\succeq \bm{D'} - \frac{1}{\bm{1}^\top\bm{D'}\bm{1}}\bm{D'}\bm{1}(\bm{D'}\bm{1})^\top,\\
    \bm{D}^{-1/2}(\bm{Sc}(G_M,V\backslash\{x\}) - \bm{L_{G_M[V\backslash \{x\}]}})\bm{D}^{-1/2}
    &\succeq \bm{D'}\bm{D}^{-1} - \bm{D}^{-1/2}\frac{1}{\bm{1}^\top\bm{D'}\bm{1}}\bm{D'}\bm{1}(\bm{D'}\bm{1})^\top\bm{D}^{-1/2}.
\end{align*}

In particular, over any subspace perpendicular to $\bm{D}^{-1/2}\bm{D'}\bm{1}$, the eigenvalue of the weighted clique is at least $\alpha/(1 + \alpha)$. As such, since the Laplacian of $G_M[V\backslash \{x\}]$ is positive semi definite, the second smallest normalised eigenvalue of the Schur complement, $\lambda_2(\bm{D}^{-1/2}(\bm{Sc}(G_M,V\backslash\{x\}))\bm{D}^{-1/2})$ is at least $\alpha/(1 + \alpha)$. 

We now apply the same technique as in Theorem~\ref{thm:main_reff}, additionally noting that since the Schur complement adds a clique to the original edges, we cannot explicitly write down the whole graph to perform random walks. We modify the random walk process by performing the random walks implicitly. 

As a preprocessing step, we first calculate the total weights of each of the edges going to $x$. Let $d''_u =  d'_u - \frac{(d'_u)^2}{d_x}$ be the part of the degree contributed by the Schur complement. We note that $d_u + d''_u$ is the new degree of the vertex $u$. While performing the random walk, at any vertex $u$, with probability $\frac{d_u}{d_u + d''_u }$, we perform a random walk step on the $O(m)$-sparse original graph $G_M[V\backslash\{x\}]$ using \textsc{UnsortedProportionalSampling}. With the remaining probability $\frac{d''_u}{d_u + d''_u}$ we perform a random walk through the weighted clique added by the Schur complement. This is easily achieved by sampling an outgoing edge proportional to the degrees $d'_v$. 

Now, since Schur complements preserve effective resistances, this gives us the effective resistance of all edges not involving $x$. 

\section{Approximate Determinants on Expanders}

Being able to more efficiently $\varepsilon$-approximate the effective resistances allows us to more efficiently approximate spanning tree counts on expanders. We prove the following more precise version of Theorem~\ref{thm:intro-thm2} in this section:
\begin{theorem}\label{thm:main_det}
    Given a graph $G$ such that the spectral gap of its normalised Laplacian is $\nu_2$, we can calculate a $(1 + \delta)$-approximation to the number of spanning trees of $G$ in $\Otil(m + n^{1.5}\delta^{-1}\nu_2^{-3})$ time with high probability.
\end{theorem}

\begin{proof}

We follow the strategy as in~\cite{DurfeePPR17}. We first begin with a description of the overall strategy as in~\cite{DurfeePPR17}.

We begin with a graph $G$. The determinant of a Laplacian is always $0$, since it has a kernel $\bm{1}$. For the sake of convenience, we abuse notation and refer to $\text{det}_+(\bm{L_G})$, the determinant of the Laplacian $\bm{L_G}$ with one of its rows and columns removed when talking about the determinant of a Laplacian. To approximate the determinant of $\bm{L_G}$, we find a $(1+\alpha)$-DD subset of vertices $V_2$ using Lemma 3.5 of \cite{KyngLPSS15}, and $V_1 = V\backslash V_2$ and recursively calculate the following:
\begin{align*}
    \text{det}_+(\bm{L_G}) = \text{det}(\bm{L_{[V_2,V_2]}})\cdot\text{det}_+(\bm{Sc}(\bm{L},V_1)).
\end{align*}

This decomposes into two parts, the determinant of a $(1+\alpha)$-DD matrix, and $\text{det}_+(\bm{Sc}(\bm{L},V_1))$. The former of these two terms is the submatrix of a Laplacian, and we can calculate it recursively by adding a new row and column, completing it into a Laplacian. 

The latter term is the determinant of a Schur complement. The Schur complement is also a Laplacian, so ideally we would like to simply recurse on this half as well. However, taking Schur complements can result in the number of edges blowing up to $\Theta(n^2)$, so an explicit construction of the Schur complement would not be fast enough. The authors of~\cite{DurfeePPR17} get around this by implicitly constructing a determinant sparsifier of the Schur complement. 

In particular, they demonstrate that a $(1 + \delta)$-determinant sparsifier of a graph can be constructed by sampling some $s$ edges of $G$, proportional to their leverage scores with $\varepsilon$ multiplicative error, and then reweighting each edge by a factor of $\exp{(\frac{n^2}{2(n-1)s})}$. For the guarantees to hold, $s$ and $\varepsilon$ are picked so that they satisfy $\frac{n^2\varepsilon^2}{s}, \frac{n^3}{s^2}\leq \delta^2$. For the best time complexity, we pick $s = n^{1.5}\delta^{-1}, \varepsilon = n^{-0.25}\delta^{0.5}$.

In fact, the bottleneck in this algorithm is precisely the time required to sample these $n^{1.5}\varepsilon^{-1}$ $(1+\varepsilon)$-approximate effective resistances, for $\varepsilon = n^{-0.25}\delta^{-0.5}$. \textcite{ChuGPSSW18} rewrites the algorithm in~\cite{DurfeePPR17}, showing the following lemma:
\begin{lemma}[\cite{DurfeePPR17, ChuGPSSW18}]\label{lem:betterschursparse}
    Let $T(m,n,s,\varepsilon)$ be the time required to find $s$ $(1 + \varepsilon)$-approximations to the effective resistances of some query edges in a graph $G$ with $m$ edges and $n$ vertices. There is an algorithm $\textsc{BetterSchurSparse}$ that takes a $2$-DD subset of vertices $V_2$, and constructs a $(1+ \delta)$-determinant sparsifier of $\bm{Sc}(G,V_1)$ with $\Otil(n^{1.5}\delta^{-1})$ edges in time:
\begin{align*}
    \Otil(m + T(m,n^{1.5}\delta^{-1},n^{-0.25}\delta^{-0.5})),
\end{align*}
and that in fact the determinant approximation algorithm \textsc{DetApprox} has the same time complexity. 

\end{lemma}

We adapt their algorithm to our effective resistance sampler that only works on $(1+\alpha)$-DD matrices or expanders. For our approach to go through, we require that taking the Schur complement preserves expansion so that our recursion holds. We first begin by proving this lemma.
\begin{lemma}\label{lem:schurincreasesnu2}
Given a graph $G$ with Laplacian $\bm{L}$ such that the second smallest eigenvalue of its normalised laplacian is $\nu_2$, for any set $C$ of vertices, the Schur complement $\bm{S} = \bm{Sc}(G,C)$ of $G$ from $F = V\backslash C$ onto $C$, normalised to the original degrees, also has second smallest eigenvalue at least $\nu_2$. 
\end{lemma}

\begin{proof}

We adopt the strategy as in the proof of the eigenvalue interlacing theorem. Let $\bm{L} = {\footnotesize
\left[ 
\begin{array}{cc} 
  \bm{F} & \bm{B} \\ 
  \bm{B^\top} & \bm{C} 
\end{array} 
\right]}$, and $\bm{D_F}^{-1/2}$ and $\bm{D_C}^{-1/2}$ be the degrees on $F$ and $C$ respectively. We have:
{\footnotesize
\begin{align}
&\left[ 
\begin{array}{cc} 
  \bm{D_F}^{-1/2} & \bm{0} \\ 
  \bm{0} & \bm{D_C}^{-1/2} 
\end{array} 
\right]
\left[ 
\begin{array}{cc} 
  \bm{F} & \bm{B} \\ 
  \bm{B}^\top & \bm{C} 
\end{array} 
\right]
\left[ 
\begin{array}{cc} 
  \bm{D_F}^{-1/2} & \bm{0} \\ 
  \bm{0} & \bm{D_C}^{-1/2} 
\end{array} 
\right]\\
=&
\left[ 
\begin{array}{cc} 
\bm{D_F}^{-1/2} & \bm{0}\\ 
\bm{0} & \bm{D_C}^{-1/2}
\end{array} 
\right]
\left[ 
\begin{array}{cc} 
\bm{I} & \bm{0}\\ 
\bm{B}^\top\bm{F}^{-1} & \bm{I}
\end{array} 
\right]
\left[ 
\begin{array}{cc} 
  \bm{F} & \bm{0} \\ 
  \bm{0} & \bm{S} 
\end{array} 
\right]
\left[ 
\begin{array}{cc} 
  \bm{I} & \bm{F}^{-1}\bm{B} \\ 
  \bm{0} & \bm{I}
\end{array} 
\right]
\left[ 
\begin{array}{cc} 
  \bm{D_F}^{-1/2} & \bm{0} \\ 
  \bm{0} & \bm{D_C}^{-1/2} 
\end{array} 
\right].
\end{align}}
Consider some vector $\bm{v} = \left[ 
{\footnotesize \begin{array}{c} 
  \bm{x} \\ 
  \bm{y}
\end{array}}
\right]$.

We have that:
{\footnotesize
\begin{align*}
\left[ 
\begin{array}{c c} 
  \bm{I} & \bm{F}^{-1}\bm{B} \\ 
  \bm{0} & \bm{I}
\end{array} 
\right]
\left[ 
\begin{array}{c c} 
  \bm{D_F}^{-1/2} & \bm{0} \\ 
  \bm{0} & \bm{D_C}^{-1/2} 
\end{array} 
\right]
\left[ 
\begin{array}{c} 
  \bm{x}\\
  \bm{y}
\end{array} 
\right]
= 
\left[ 
\begin{array}{c} 
  \bm{D_F}^{-1/2}\bm{x} + \bm{F}^{-1}\bm{B}\bm{D_C}^{-1/2}\bm{y}\\
  \bm{D_C}^{-1/2}\bm{y}
\end{array} 
\right].
\end{align*}}

Now for any $\bm{y}$, consider its extension into the whole space $v(\bm{y}) = \left[ 
{\footnotesize \begin{array}{c} 
  -\bm{D_F}^{1/2}\bm{F}^{-1}\bm{B}\bm{D_C}^{-1/2}\bm{y} \\ 
  \bm{y}
\end{array}}
\right]$. Notice that $\bm{F}^{-1}$,$\bm{D_C}^{-1/2}$ and $\bm{D_F}^{-1/2}$ are all well defined, since $\bm{F}$ is positive definite, while $\bm{D_C}$ and $\bm{D_F}$ are diagonal and non-zero. We also have that ${\footnotesize
\left[ 
\begin{array}{c c} 
  \bm{I} & \bm{F}^{-1}\bm{B} \\ 
  \bm{0} & \bm{I}
\end{array} 
\right]
\left[ 
\begin{array}{c c} 
  \bm{D_F}^{-1/2} & \bm{0} \\ 
  \bm{0} & \bm{D_C}^{-1/2} 
\end{array} 
\right]
v(\bm{y})
= 
\left[ 
\begin{array}{c} 
  \bm{0}\\
  \bm{D_C}^{-1/2}\bm{y}
\end{array}
\right]}$, which combined with equation (2) gives us that $\bm{v}(\bm{y})^\top\bm{D}^{-1/2}\bm{LD}^{-1/2}\bm{v}(\bm{y}) = \bm{y}^\top\bm{D_C}^{-1/2}\bm{S}\bm{D_C}^{-1/2}\bm{y}$. This gives us a bijection from the quadratic form of the normalised Laplacian $\bm{D}^{-1/2}\bm{LD}^{-1/2}$ to the quadratic form of the Schur complement $\bm{D_C}^{-1/2}\bm{SD_C}^{-1/2}$, normalised to the old degrees. 

Now let $\bm{v_1}$ and $\bm{v_2}$ be the smallest two eigenvectors of $\bm{D_C}^{-1/2}\bm{SD_C}^{-1/2}$, and let $W = \text{span}\{\bm{v_1},\bm{v_2}\}$. We have:
\begin{align*}
    \lambda_2(\bm{D_C}^{-1/2}\bm{SD_C}^{-1/2}) &= \max_{\bm{y} \in W}\frac{\bm{y}^\top\bm{D_C}^{-1/2}\bm{SD_C}^{-1/2}\bm{y}}{\bm{y}^\top\bm{y}}\\
    &= \max_{\bm{y} \in W}\frac{\bm{v}(\bm{y})^\top\bm{D}^{-1/2}\bm{LD}^{-1/2}\bm{v}(\bm{y})}{\bm{y}^\top\bm{y}}\\
    &\geq \max_{\bm{y} \in W}\frac{\bm{v}(\bm{y})^\top\bm{D}^{-1/2}\bm{LD}^{-1/2}\bm{v}(\bm{y})}{\bm{v}(\bm{y})^\top\bm{v}(\bm{y})}.
\end{align*}

But if $\bm{y}$ is in the two dimensional subspace $W$, $\bm{v}(\bm{y})$ also lies in some two dimensional vector subspace $W' = \bm{v}(W)$. To see this, notice that $W'$ is indeed a vector subspace since $\bm{v}(\bm{x}+\bm{y}) = \bm{v}(\bm{x})+\bm{v}(\bm{y})$. Also, $W'$ is exactly $2$-dimensional, since it can be spanned by $\bm{v}(\bm{w_1}),v(\bm{w_2})$ and $\bm{v}(\bm{w_1})\neq c\bm{v}(\bm{w_2})$, for any basis $\bm{w_1},\bm{w_2}$ of $W$.

Hence, we have that:
\begin{align*}
    \lambda_2(\bm{D_C}^{-1/2}\bm{SD_C}^{-1/2}) &\geq \max_{\bm{y} \in W}\frac{\bm{v}(\bm{y})^\top\bm{D}^{-1/2}\bm{LD}^{-1/2}\bm{v}(\bm{y})}{\bm{v}(\bm{y})^\top\bm{v}(\bm{y})}\\
    &= \max_{\bm{v} \in v(W)}\frac{\bm{v}^\top\bm{D}^{-1/2}\bm{LD}^{-1/2}\bm{v}}{\bm{v}^\top\bm{v}}\\
    &\geq \min_{\text{dim}(U) = 2}\max_{\bm{v} \in U} \frac{\bm{v}^\top\bm{D}^{-1/2}\bm{LD}^{-1/2}\bm{v}}{\bm{v}^\top\bm{v}}\\
    &= \lambda_2(\bm{D^{-1/2}LD^{-1/2}}). 
\end{align*}

\end{proof}

A consequence of this lemma is that if $G$ has second smallest normalised eigenvalue $\nu_2$, then the Schur complement $\bm{Sc}(G,S)$, normalised to its new degrees also has second smallest normalised eigenvalue at least $\nu_2$. Since the new degrees of the Schur complement, represented by say the matrix $\bm{D_S}$, are entrywise smaller than $\bm{D_C}$, we have that $\lambda_2(\bm{D_S}^{-1/2}\bm{SD_S}^{-1/2}) \geq \lambda_2(\bm{D_C}^{-1/2}\bm{SD_C}^{-1/2})$. 

Let $\bm{f}(\bm{x}) = \bm{D_C}^{1/2}\bm{D_S}^{-1/2}\bm{x}$. Notice that the map $\bm{f}$ forms a bijection between vector spaces of dimension 2, since $\bm{f}(U)$ is 2-dimensional if $U$ is 2-dimensional, and $\bm{f}$ is invertible. We also have that $|\bm{f}(\bm{x})| \geq |\bm{x}|$, since it is entrywise larger. Hence:
\begin{align*}
    \lambda_2(\bm{D_S}^{-1/2}\bm{SD_S}^{-1/2}) &= \min_{\text{dim}(U) = 2}\max_{\bm{x} \in  U}\frac{\bm{x}^\top\bm{D_S}^{-1/2}\bm{SD_S}^{-1/2}\bm{x}}{\bm{x}^\top\bm{x}}\\
    &\geq \min_{\text{dim}(U) = 2}\max_{\bm{x} \in U}\frac{\bm{x}^\top\bm{D_S}^{-1/2}\bm{SD_S}^{-1/2}\bm{x}}{\bm{f}(\bm{x})^\top\bm{f}(\bm{x})}\\
    &= \min_{\text{dim}(U) = 2}\max_{\bm{x} \in U}\frac{\bm{f}(\bm{x})^\top\bm{D_C}^{-1/2}\bm{SD_C}^{-1/2}\bm{f}(\bm{x})}{\bm{f}(\bm{x})^\top\bm{f}(\bm{x})}\\
    &=\lambda_2(\bm{D_C}^{-1/2}\bm{SD_C}^{-1/2}).
\end{align*}

The algorithm in~\cite{DurfeePPR17} also constructs determinant sparsifiers. For our recursive guarantees to hold, we also have to demonstrate that spectral gap of the normalised Laplacian can not change that much after the sketching process. 
\begin{lemma}\label{lem:detsparsifypreservesnu2}
Consider a graph $G$ so that the spectral gap of its normalised Laplacian is $\nu_2$. Let $H$ be the graph produced by determinant sparsifier $\textsc{DetSparsify}$ in~\cite{DurfeePPR17} with some $s = n^{1.5}\delta^{-1}$ edges. Then, the spectral gap of the normalised Laplacian for $H$ is at least some $(1 - O(\delta^{1/4}))\nu_2$. 
\end{lemma}

\begin{proof}
Note that DetSparsify (Algorithm 2 of \cite{DurfeePPR17}) essentially samples $s = n^{1.5}\delta^{-0.5}$ edges proportional to the edge leverage scores, and rescaling by a factor of $\exp(\frac{n^2}{2(n-1)s})$. As such, the graph constructed, with weights rescaled back, ie $\exp(-\frac{n^2}{2(n-1)s})H$ is a $\delta^{1/4}$-spectral sparsifier of $G$ (See Theorem 1 of \cite{SpielmanS08}).

Let $\bm{f}(\bm{x}) = \bm{D_G}^{1/2}\bm{D_H}^{-1/2}\bm{x}$. $\bm{f}$ is a mapping that forms a bijection between two dimensional subspaces $U \rightarrow \bm{f}(U)$. Since $H$ is a $\delta^{1/4}$-spectral sparsifier of $G$, their degrees differ by a factor of at most $1 + \delta^{1/4}$, and as such $\bm{f}(\bm{x})^\top\bm{f}(\bm{x}) \geq (1 - 2\delta^{1/4})\bm{x}^\top\bm{x}$. We have:
\begin{align*}
    \lambda_2(\bm{D_H}^{-1/2}\bm{HD_H}^{-1/2}) &= \min_{\text{dim}(U) = 2}\max_{\bm{x} \in  U}\frac{\bm{x}^\top\bm{D_H}^{-1/2}\bm{HD_H}^{-1/2}\bm{x}}{\bm{x}^\top\bm{x}}\\
    &= \min_{\text{dim}(U) = 2}\max_{\bm{x} \in  U}\frac{\bm{f}(\bm{x})^\top\bm{D_G}^{-1/2}\bm{HD_G}^{-1/2}\bm{f}(\bm{x})}{\bm{x}^\top\bm{x}}\\ 
    &\geq (1 - 2\delta^{1/4})\min_{\text{dim}(U) = 2}\max_{\bm{x} \in  U}\frac{\bm{f}(\bm{x})^\top\bm{D_G}^{-1/2}\bm{HD_G}^{-1/2}\bm{f}(\bm{x})}{\bm{f}(\bm{x})^\top\bm{f}(\bm{x})}\\ 
   &\geq (1 - 3\delta^{1/4})\min_{\text{dim}(U) = 2}\max_{\bm{x} \in  U}\frac{\bm{f}(\bm{x})^\top\bm{D_G}^{-1/2}\bm{GD_G}^{-1/2}\bm{f}(\bm{x})}{\bm{f}(\bm{x})^\top\bm{f}(\bm{x})}\\ 
   &= \lambda_2(\bm{D_G}^{-1/2}\bm{GD_G}^{-1/2}).
\end{align*}
\end{proof}

We now demonstrate that we can solve determinants on expanders. We use the same recursive strategy as in ~\cite{DurfeePPR17}. We begin with a graph $G$ with spectral gap of the normalised Laplacian being $\nu_2$. To calculate the determinant of $\bm{L_G}$, we find a $2$-DD subset of vertices $V_2$ using Lemma 3.5 of \cite{KyngLPSS15}, and $V_1 = V\backslash V_2$ and recursively calculate the following:
\begin{align*}
    \text{det}_+(\bm{L_G}) = \text{det}(\bm{L_{[V_2,V_2]}})\cdot\text{det}_+(\bm{Sc}(\bm{L},V_1)).
\end{align*}

This decomposes into two parts, the determinant of a $2$-DD matrix, and $\text{det}_+(\bm{Sc}(\bm{L},V_1))$, defined to be the determinant of the Schur complement onto the rest of the vertices with a row and column removed. The entirety of the algorithm is the same, with the only difference being the subroutine used to approximate effective resistances. This leads to the following differences:
\begin{enumerate}
    \item In~\cite{DurfeePPR17}, $\text{det}(\bm{L_{[V_2,V_2]}})$ is calculated by simply completing $\bm{L_{[V_2,V_2]}}$ into a Laplacian and recursing. Since our algorithm involves expansion, for our recursion to hold, we have to demonstrate that this new graph, with all of $V_1$ being contracted to a single vertex, does not have $\nu_2$ that decreases by a large amount. Instead of doing this, we slightly modify the recursion, and use Lemma~\ref{lem:2ddreff}. 
    \item In the second half of the recursion, $\text{det}_+(\bm{Sc}(\bm{L},V_1))$ can potentially have too many edges if explicitly represented, so the determinant is calculated by first constructing a determinant sparsifier, and then recursing. If there were no sketches involved, Lemma~\ref{lem:schurincreasesnu2} would tell us that taking Schur complements can only increase $\nu_2$, allowing us to recurse. However, we have to demonstrate that the sketching process does not reduce $\nu_2$ by too much. 
\end{enumerate}

\subsection{$2$-DD matrix}

Let $\bm{L^{V_2}}$ be the completion of $\bm{L_{[V_2,V_2]}}$ into a Laplacian, and let the new vertex be $x$. By removing the single vertex $x$, we have the following:
\begin{align*}
\text{det}(\bm{L_{[V_2,V_2]}}) &= \text{det}_+(\bm{L^{V_2}})\\
&= \deg(x)\cdot\text{det}_+(\bm{Sc}(\bm{L^{V_2}},V_2)).
\end{align*}

By Lemma~\ref{lem:2ddreff}, we have that we can calculate the effective resistances on all edges of the $2$-DD matrix that are not adjacent to the new vertex $x$, which gives us the effective resistances of the edges in $\bm{Sc}(\bm{L^{V_2}},V_2)$. This allows us to construct a determinant preserving sparsifier of the Schur complement $H^{V_2}$ using SchurSparse as in \cite{DurfeePPR17}, taking total time $\Otil(m + n\varepsilon^{-2})$ to calculate the effective resistances, and $\Otil(s)$ time to sample $s$ edges. In particular we require choices of $s$ and $\varepsilon$ such that $\frac{n^2\varepsilon^2}{s}, \frac{n^3}{s^2}\leq \delta^2$, for the guarantees of SchurSparse to hold, so we pick $s = n^{1.5}\delta^{-1}, \varepsilon = n^{-0.25}\delta^{0.5}$, giving us total time $\Otil(s + n\varepsilon^{-2}) = \Otil(n^{1.5}\delta^{-1})$. 

We now have a sketch $H^{V_2}$ of size $n^{1.5}\delta^{-1}$. From Lemma~\ref{lem:2ddreff}, we know that $\bm{Sc}(\bm{L^{V_2}},V_2)$ is a graph with expansion at least some constant. Lemma~\ref{lem:detsparsifypreservesnu2} then guarantees that after the sketching process, the $\nu_2$ decreases by at most some $(1 - O(\delta^{1/4}))$ factor. This guarantees that $H^{V_2}$ is also a graph with constant expansion, allowing us to continue by recursing. We are essentially picking a specific 2-DD set to recurse on, so error guarantees still hold through a similar argument as in \cite{DurfeePPR17}. We see below in Figure~\ref{fig:recursive_structure} the new recursive structure. \cite{DurfeeKPRS17} prove that the total error in each layer is sufficiently small. We note that the fact that the choice of subset was not important in the error analysis, so we have the same error guarantees. As a result of our recursion on the right half reducing the problem size by only 1 as opposed to a factor of 2, our recursion will have twice the depth, but this also does not affect error guarantees. 

\begin{figure}
\caption{Recursive structure of DetApprox from \cite{DurfeePPR17}, with overlaid changes}
\centering
\includegraphics[width=\textwidth]{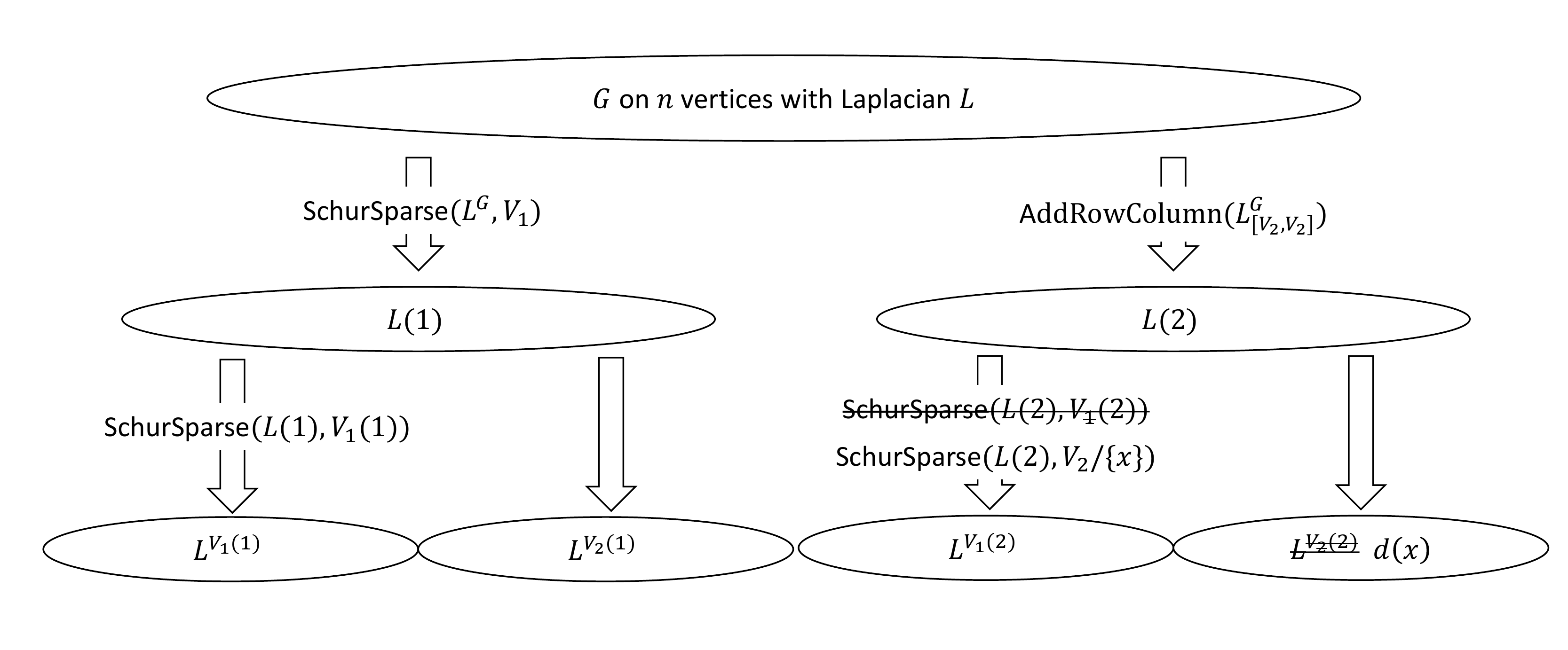}
\label{fig:recursive_structure}
\end{figure}

\subsection{Schur Complement}\label{sect_detschur}

As in the case of~\cite{DurfeePPR17}, we do not have to calculate effective resistances on the Schur complement. Since Schur complements preserve effective resistances, we can instead calculate the effective resistances on $\bm{L}$ instead of $\bm{Sc}(\bm{L^{V_2}},V_1)$.  

As such, the only difference is in maintaining expansion guarantees through the sketching process. Again, just as in the $2$-DD case, by Lemma~\ref{lem:detsparsifypreservesnu2}, the determinant sparsifier produced $H^{V_2}$ is a rescaled spectral sparsifier of $\bm{Sc}(\bm{L^{V_2}},V_1)$, and as such the sampled graph $H^{V_2}$'s second smallest normalised eigenvalue can decrease by at most a factor of $(1 - 2\delta^{1/4})$. Suppose the graph at the topmost level has spectral gap of the normalised Laplacian $\nu_2$. The resulting graph also always has spectral gap at least some $\Otil(\nu_2)$, since its expansion drops by at mmost a factor of $2\delta^{1/4}$, or gets set back to a constant value in the case of it being generated from the $2$-DD component. 

Putting this all together, we see that at each level of our recursion, the costs are dominated by the time required to produce the effective resistance estimates, taking $\Otil(s + n\varepsilon^{-2}\nu_2^{-3}) = \Otil(n^{1.5}\delta^{-1}\nu_2^{-3})$ time. Given that there are only $O(\log{n})$ layers, the total time complexity of the algorithm comes out to $\Otil(m + n^{1.5}\delta^{-1}\nu_2^{-3})$. \end{proof}

A corollary of the above method is that we can calculate the determinant of any $(1 + \alpha)$-diagonally dominant matrix $\bm{M}$ that is a submatrix of a Laplacian (Theorem~\ref{thm:intro-thm3}).
\begin{proof}\emph{(of~Theorem~\ref{thm:intro-thm3})}
The proof of Lemma~\ref{lem:2ddreff} gives us that the Schur complement of $\bm{L_M}$ onto $V\backslash \{x\}$ has $\nu_2$ at least $\alpha/(1 + \alpha)$. Using the same recursive step as in Section~\ref{sect_detschur}, we essentially work on a graph with $\nu_2$ at least $\alpha/(1 + \alpha)$, giving us a time complexity of $\Otil(m + n^{1.5}\delta^{-1}(1 + \alpha^{-3}))$. 
\end{proof}

\pagebreak

\printbibliography

\pagebreak

\appendix

\section{Deferred Proofs}

We provide some proofs of relatively standard facts here. We first prove a lower bound on the effective resistances.

\factrefflowerbound*

\begin{proof}
For any vector $\bm{x}$, we have that the quadratic form $\bm{x}^\top(\bm{D} + \bm{A})\bm{x}$ can be written as:
\begin{align*}
    \bm{x}^\top(\bm{D} + \bm{A})\bm{x} = \sum_{(u,v) \in E}w_{uv}(\bm{x}_u + \bm{x}_v)^2.
\end{align*}
Thus, $\bm{D} + \bm{A} \succeq \bm{0},$ and $\bm{L} = \bm{D} - \AA \preceq 2 \DD.$ 
Hence, for vectors perpendicular to $\bm{1},$ we have
$    \frac{1}{2}\bm{D}^{-1} \preceq \bm{L}^+.$
This gives us:
\begin{align*}
    R(u,v) &= (\bm{1_u} - \bm{1_v})^\top\bm{L}^+(\bm{1_u} - \bm{1_v})\\
    &\geq (\bm{1_u} - \bm{1_v})^\top\frac{1}{2}\bm{D}^{-1}(\bm{1_u} - \bm{1_v})\\
    &= \frac{1}{2}(\frac{1}{d_u}+\frac{1}{d_v}).
\end{align*}

An alternative proof is as follows. By Rayleigh's Monotonicity Law, the effective resistance between $u$ and $v$ will only decrease if the resistances of edges are lowered. WLOG, let $d_u < d_v$. Construct the graph $G'$ where we decrease the resistance of every single edge not adjacent to $u$ to $0$, noting that by Rayleigh's Monotonocity Law, the resistance of $u$ and $v$ in this graph is a lower bound for the resistance between $u$ and $v$ in $G$. Notice now we can collapse all vertices that are not $u$ into a single vertex, since they are all connected by $0$ resistance edges, resulting in a circuit with parallel edges from $u$ to $v$. This gives us that the effective resistance between $u$ and $v$ is just $\frac{1}{d_u}$, lower bounding the effective resistance between any two vertices by $\max \{\frac{1}{d_u},\frac{1}{d_v}\}$. 

\end{proof}

Next we prove a result on the convergence of random walks.

\lemrandomwalk*

\begin{proof}
Let $\bm{N} = \bm{D}^{-1/2}\bm{L}\bm{D}^{-1/2}$ be the normalized Laplacian. Since this matrix is symmetric, it has an orthogonal basis of eigenvectors. Let $\bm{\psi_1}, \ldots, \bm{\psi_n}$ denote an orthonormal eigenbasis for $\NN$ with eigenvalues $\nu_1 \le \ldots \le \nu_n$ respectively.  
We have that:
\begin{align*}
    \frac{1}{2}\bm{I} + \frac{1}{2}\bm{AD}^{-1} = \bm{I} - \frac{1}{2}\bm{D}^{1/2}\bm{N}\bm{D}^{-1/2},
\end{align*}
which gives us that for each $i$ if $,\bm{D}^{1/2}\bm{\psi_i}$ is an eigenvector of  $\frac{1}{2}\bm{I} + \frac{1}{2}\bm{AD}^{-1}$ with eigenvalue  $(1 - \frac{1}{2}\nu_i).$
Given an initial distribution $\bm{p}$, we write $\bm{D}^{-1/2}\bm{p}$ as a linear combination of the orthogonal eigenvectors of $\bm{N},$ i.e., for some $\alpha_i \in \rea,$ we have $   \bm{D}^{-1/2}\bm{p} = \sum_i \alpha_i \bm{\psi_i}.$ Thus, $    \bm{p} = \sum_i \alpha_i \bm{D}^{1/2}\bm{\psi_i},$ and,
\begin{align*}
    (\frac{1}{2}\bm{I} + \frac{1}{2}\bm{AD}^{-1})\bm{p} &= \sum_i \alpha_i(1 - \frac{1}{2}\nu_i)\bm{D}^{1/2}\bm{\psi_i}.
\end{align*}
In particular, we have that $\nu_1 = 0$ and  $\bm{\psi_1} = \frac{1}{\bm{1}^{\top} \DD \bm{1}}\bm{D}^{1/2}\bm{1}.$
Thus, its coefficient $\alpha_1$ is
\begin{align*}
\alpha_1 = \bm{\psi_1}^\top\sum_i \alpha_i \bm{\psi_i} = \bm{\psi_1}^\top\bm{D}^{-1/2}\bm{p} = 
\frac{1}{\bm{1}^{\top} \DD \bm{1}}\bm{1}^{\top} \pp
= \frac{1}{\bm{1}^\top\bm{D}\bm{1}},
\end{align*}

We bound each individual term after $t$ steps. Let $\bm{\pi}$ be the stationary distribution of the lazy random walk, $(\bm{\pi})_u = \frac{d_u}{\sum_u d_u}.$ Thus $\ppi = \frac{1}{\bm{1}^{\top}\DD \bm{1}}{\DD\bm{1}} = \alpha_1\bm{D}^{1/2}\psi_1$. We have that the $\ell_1$ distance between the stationary distribution and the random walk distribution after $t$ steps is bounded by: 
\begin{align*}
    \norm{(\frac{1}{2}\bm{I} + \frac{1}{2}\bm{AD}^{-1})^{t}\bm{p} - \bm{\pi}}_1 &= \norm{\sum_i \alpha_i(1 - \frac{1}{2}\nu_i)^{t}\bm{D}^{1/2}\bm{\psi_i} - \bm{\pi}}_1\\
    &= \norm{\sum_{i > 1} \alpha_i(1 - \frac{1}{2}\nu_i)^{t}\bm{D}^{1/2}\bm{\psi_i}  + \alpha_1\bm{D}^{1/2}\psi_1 - \bm{\pi}}_1\\
    &= \norm{\sum_{i > 1} \alpha_i(1 - \frac{1}{2}\nu_i)^{t}\bm{D}^{1/2}\bm{\psi_i}}_1\\
    &= \norm{\bm{D}^{1/2}\sum_{i > 1} \alpha_i(1 - \frac{1}{2}\nu_i)^{t}\bm{\psi_i}}_1
\end{align*}
Applying Cauchy-Schwarz, we have that this is bounded by,
\begin{align*}
    \norm{(\frac{1}{2}\bm{I} + \frac{1}{2}\bm{AD}^{-1})^{t}\bm{p} - \bm{\pi}}_1
    & \le \norm{\DD^{1/2}\bm{1}}_2 \norm{\sum_{i > 1} \alpha_i(1 - \frac{1}{2}\nu_i)^{t}\bm{\psi_i}}_2 \\
    & \le \sqrt{\bm{1}^\top \DD \bm{1}} \sqrt{\sum_{i > 1} \alpha_i^2 (1 - \frac{1}{2}\nu_i)^{2t}} \\
    & \le e^{-t\nu_2/2} \sqrt{\bm{1}^\top \DD \bm{1}} \sqrt{\sum_{i > 1} \alpha_i^2} \\
    & \le e^{-t\nu_2/2} \sqrt{\bm{1}^\top \DD \bm{1}} \sqrt{\pp^{\top}\DD^{-1} \pp} \\
    & \le e^{-t\nu_2/2} \frac{nd_{\max}}{d_{\min}},
\end{align*}
where the last two inequalities follow from $\sum_i \alpha_i^2 = \pp^{\top}\DD^{-1} \pp \le 1.$

\end{proof}

\end{document}